\documentclass{ifacconf}

\usepackage{graphicx}
\usepackage{natbib}
\usepackage{amsmath,amssymb,amsfonts}
\usepackage{mathtools}
\usepackage{bm}
\usepackage{booktabs}
\usepackage{tikz}
\usetikzlibrary{automata,positioning,arrows.meta,shapes.misc,calc}
\usepackage{pifont}
\newcommand{\xmark}{\ding{55}}%

\newtheorem{theorem}{Theorem}
\newtheorem{lemma}[theorem]{Lemma}

\newtheorem{definition}{Definition}
\newtheorem{remark}{Remark}
\newtheorem{assumption}{Assumption}
\newtheorem{condition}{Condition}

\newcommand{\R}{\mathbb{R}}
\newcommand{\N}{\mathbb{N}}

\newcommand{\E}{\mathbb{E}}
\newcommand{\Prob}{\mathbb{P}}
\newcommand{\Ltwo}{\mathcal{L}_2}

\newcommand{\xtilde}{\tilde{x}}

\DeclareMathOperator{\He}{He}

\newif\ifFullVersion
\FullVersiontrue  

\begin{document}
\begin{frontmatter}

\ifFullVersion
  \title{Expected String Stability of Human-Led Vehicle Platoons under Stochastic Communication Delays (Full Version)}
\fi

\thanks[footnoteinfo]{This work was funded by the Chilean National Agency for Research
and Development (ANID) through the grants FONDECYT Initiation 11250715, FONDECYT Regular 1251406, and ANID CPS-RTC CIA250016.}

\author[A]{Francisco Aguilera}
\author[A]{Víctor Jaque}
\author[UAI]{Andrés A. Peters}
\author[A,B]{Alejandro I. Maass}

\address[A]{Department of Electrical Engineering, Pontificia Universidad Cat\'olica de Chile,  Santiago, 7820436, Chile. (email: \textnormal{\small \texttt{\{francisco.aguilera,victor.jaque,alejandro.maass\}@uc.cl)}}}
\address[UAI]{Faculty of Engineering and Sciences, Universidad Adolfo Ib\'a\~{n}ez, Pe\~{n}alol\'en, 7941169, Santiago, Chile. (email: \textnormal{\small \texttt{andres.peters@uai.cl)}}}  
\address[B]{Cyber-Physical Systems Research and Technology Center, Pontificia Universidad Cat\'olica de Chile,  Santiago, 7820436, Chile}

\begin{abstract} 
This paper studies expected $\mathcal{L}_2$ string stability of event-triggered vehicle platoons in which a human driver leads a chain of cooperatively controlled autonomous followers under stochastic communication delays. The leader's driving behavior propagates through the string via vehicle-to-vehicle (V2V) communication, so human-induced disturbances must not amplify along the platoon. Unlike deterministic approaches based on worst-case delay bounds, we derive string-stability conditions depending on the full delay distribution through integral inequalities. The closed-loop platoon is modeled as a stochastic hybrid system capturing vehicle dynamics, communication events, and event-triggering. This framework certifies string stability even when delays exceed deterministic admissible bounds with nonzero probability. Results are evaluated under several delay distributions using the MATLAB HyEQ simulator.
\end{abstract}

\end{frontmatter}

\section{Introduction}
\label{sec:intro}

Vehicle platooning is a representative cyber-physical human system that integrates sensing, computation, communication, and control to coordinate a human-driven leader with a chain of autonomous followers in real time \cite{jia2015survey}. By enabling short inter-vehicle distances, platoons can improve road capacity, fuel efficiency, and safety. However, safe operation requires more than individual vehicle stability: the human leader's driving behavior (acceleration and braking maneuvers) must not amplify as disturbances along the autonomous string, a property known as string stability \cite{feng2019string}.

Cooperative adaptive cruise control (CACC) achieves string-stable platooning at small inter-vehicle distances using vehicle-to-vehicle (V2V) communication to transmit signals such as acceleration from preceding vehicles \cite{milanes2013cooperative}. However, the resulting network-induced imperfections, especially nonzero communication delays, distort the leader's driving signals as they propagate through the string and can compromise string stability, and must therefore be explicitly considered in CACC analysis and design \cite{zhang2023impacts}.

String stability under delayed V2V communication has been extensively studied under deterministic assumptions, yielding conditions for constant or bounded delays; see, e.g., \citep{gao2016robust,Dolk2017,elahi2024fixed,samii2024simultaneous}. However, in practical wireless networks, delays are inherently stochastic and may occasionally exceed deterministic design thresholds, making worst-case delay bounds overly conservative and motivating the study of string stability under stochastic communication delays.

Only a few works have addressed this setting. Early studies \citep{Qin2014,Qin2015ACC} considered connected cruise control with intermittent V2V communication, packet losses, and stochastic delays, with \citep{Qin2015ACC} deriving mean and covariance dynamics to analyze plant and string stability. However, delays are modeled as integer multiples of a fixed sampling period, restricting the analysis to discrete delay distributions. The work \citep{Wang2020CCC} proposed an LQ-based optimal controller that accounts for delays, packet losses and state noise while reducing headway and velocity deviations, but does not provide string stability guarantees.  More recently, \cite{pan2026string} derived delay-distribution-dependent conditions for non-uniform stochastic delays, where the delay law enters only through an interval-based abstraction: the range is partitioned into intervals and Bernoulli variables encode membership, so the analysis depends on the distribution solely through interval probabilities. In contrast, we let the distribution enter the certificate directly through an integral with respect to its probability measure.

In this paper, we derive string stability conditions that depend on the delay distribution through integral conditions with respect to its probability measure. Adopting a hybrid systems formalism, the analysis captures continuous-time vehicle dynamics and discrete events induced by transmissions, receptions, and event-triggering. Building on \cite{Dolk2017} and the stochastic hybrid systems framework of \cite{Schlotterbeck2024}, we relax the worst-case MAD (Maximum Allowable Delay) requirement and allow delays to exceed the deterministic MAD with nonzero probability while preserving stability guarantees, ensuring that human-induced disturbances from the leader are safely attenuated across the autonomous string even under realistic stochastic network conditions.

\emph{Notation:} Let $\R$, $\R_{\ge 0}$, and $\N$ denote the real,
nonnegative real, and nonnegative integer sets, respectively, and let
$\|\cdot\|$ be the Euclidean norm. For symmetric $M$,
$M \succ 0$ ($M \succeq 0$) and $M \prec 0$ ($M \preceq 0$) denote
positive and negative definiteness (semidefiniteness). For a matrix $A$,
$\He(A):=A+A^\top$. For $w:\R_{\ge 0}\to\R^n$,
$\|w\|_{\Ltwo,[0,T]}^2 := \int_0^T \|w(t)\|^2\,dt$.
The expectation and probability operators are $\E[\cdot]$ and
$\Prob(\cdot)$. For a probability measure $\mu$,
$\int f(v)\,\mu(dv)=\E[f(V)]$ for $V\sim\mu$.

\section{Problem setting}
\label{sec:problem_setting}


\begin{figure}
    \centering
    \includegraphics[width=0.95\linewidth]{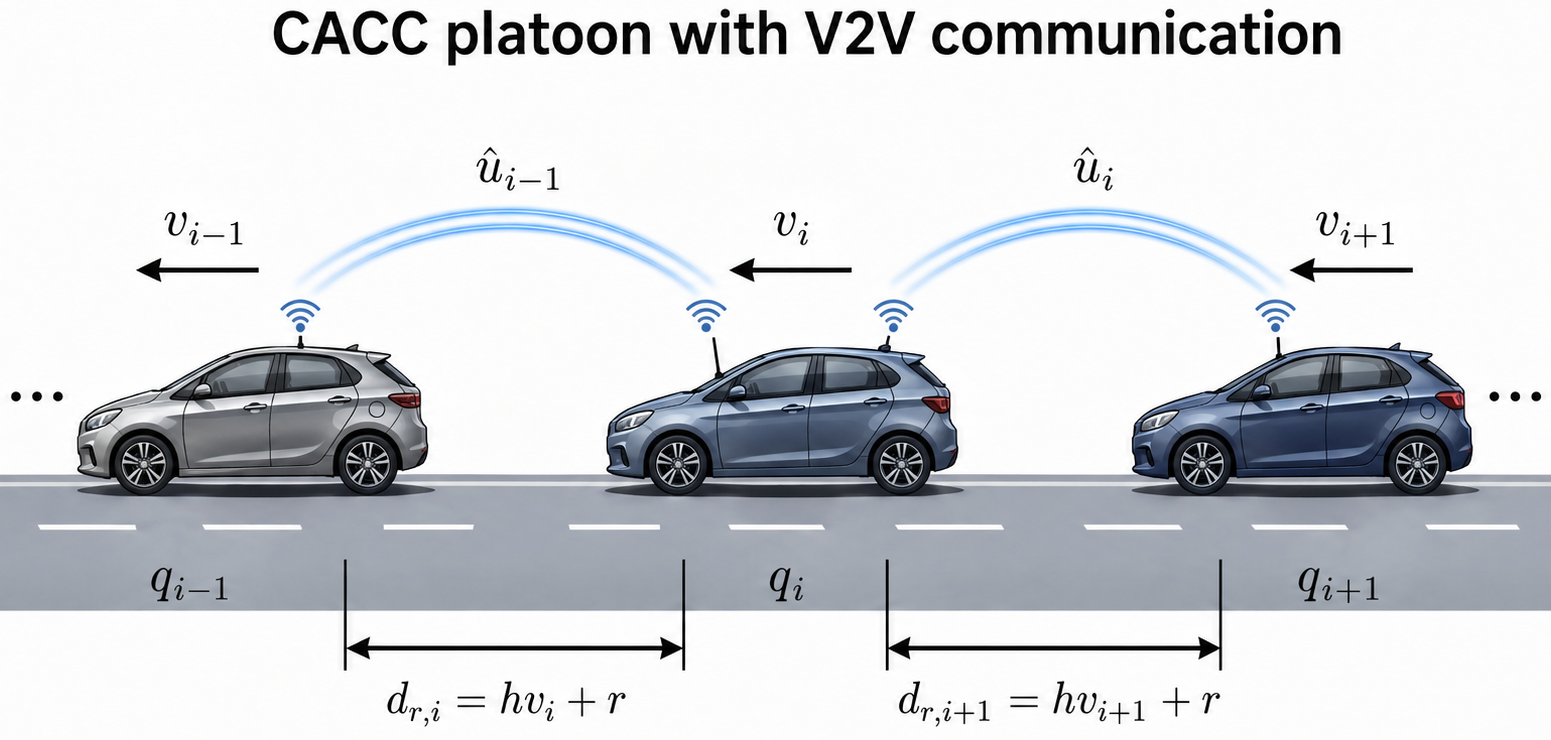}
    \caption{Schematic of the predecessor-following platoon model.}
    \label{fig:platoon-block}
\end{figure}
\subsection{CACC closed-loop platoon model}

Consider a platoon of $N+1$ vehicles indexed $i = 0, 1, \ldots, N$, depicted in Fig.~\ref{fig:platoon-block}, where vehicle~$0$ is a human-driven leader and each follower $i \geq 1$ implements CACC using on-board sensors and the predecessor's acceleration received via V2V~\citep{Dolk2017}. Each vehicle has a first-order actuator (time constant $\tau_d > 0$) and a velocity-damping coefficient $\alpha \geq 0$ aggregating rolling resistance and aerodynamic drag about the cruise speed:
\begin{equation}
    \dot{v}_i = -\alpha v_i + a_i, \qquad \dot{a}_i = -\frac{1}{\tau_d} a_i + \frac{1}{\tau_d} u_i.
    \label{eq:vehicle_dynamics}
\end{equation}
Including this coefficient yields a more realistic model than the pure integrator of~\cite{Dolk2017, Ploeg2014}, recovered when $\alpha = 0$. A constant time-gap policy with headway $h > 0$ defines the desired inter-vehicle distance $d_{r,i}(t) := h\,v_i(t) + r$, and the spacing error is
\begin{equation}
    e_i := q_{i-1} - q_i - L_i - d_{r,i}, \quad
    \dot{e}_i = (v_{i-1} - v_i) - h\dot{v}_i,
    \label{eq:spacing_error}
\end{equation}
where $q_i$ is the position of vehicle~$i$ and $L_i$ its length. Substituting~\eqref{eq:vehicle_dynamics} into~\eqref{eq:spacing_error}, we can write
\begin{equation}
    \dot{e}_i = v_{i-1} - v_i(1 - h\alpha) - ha_i.
    \label{eq:edot}
\end{equation}
%
In terms of controller, we retain the widely adopted CACC structure from~\cite{Dolk2017, Ploeg2014} 
\begin{equation}
    \chi_i = k_p e_i + k_d \dot{e}_i + \hat{u}_{i-1}, \quad 
    \dot{u}_i = -\frac{1}{h}u_i + \frac{1}{h}\chi_i
    \label{eq:controller}
\end{equation}
with gains $k_p, k_d > 0, \quad k_d-k_p\tau_d>0$ (\cite{Dolk2017}) and predecessor desired acceleration 
$\hat{u}_{i-1}$ sent via V2V communication subject to delays. 

Since the stochastic delays affect the dynamics of the transmitted variable $\hat{u}_{i-1}$, a useful variable to define is the so-called \emph{network-induced error} defined as
\begin{equation*}
    e_{u_{i-1}} :=\hat{u}_{i-1}-u_{i-1}. \label{eq:network_induced_error}
\end{equation*}
The received acceleration $\hat{u}_{i-1}$ is held between samples via a zero-order hold (ZOH), i.e., $\dot{\hat{u}}_{i-1} = 0$. The leader vehicle $i = 0$ is operated by a human driver whose acceleration input is available instantaneously to their own vehicle, i.e.\ $\hat{u}_0 = u_0$, so no network-induced error is present.

We define the state for vehicle pair $(i-1,\,i)$:
\begin{equation*}
    \tilde{x}_i := \bigl[v_{i-1},\, a_{i-1},\, u_{i-1},\, e_i,\, v_i,\, a_i,\, u_i\bigr]^T \in \mathbb{R}^7.
    \label{eq:state}
\end{equation*}
The closed-loop matrices are given by $A_{11} := A + EC$, $A_{12} := E$, $A_{13} := B$,
and combining~\eqref{eq:vehicle_dynamics}--\eqref{eq:controller}, the augmented dynamics under CACC are
\begin{equation*}
    \dot{\tilde{x}}_i = A_{11}\tilde{x}_i + A_{12}e_{u_{i-1}} + 
    A_{13}\chi_{i-1},
    \label{eq:cl_dynamics}
\end{equation*}
where 
\begin{equation*}
A =
\begin{bmatrix}
-\alpha & 1 & 0 & 0 & 0 & 0 & 0 \\
0 & -1/\tau_d & 1/\tau_d & 0 & 0 & 0 & 0 \\
0 & 0 & -1/h & 0 & 0 & 0 & 0 \\
1 & 0 & 0 & 0 & -(1-h\alpha) & -h & 0 \\
0 & 0 & 0 & 0 & -\alpha & 1 & 0 \\
0 & 0 & 0 & 0 & 0 & -1/\tau_d & 1/\tau_d \\
k_d/h & 0 & 0 & k_p/h & -k_d(1/h-\alpha) & -k_d & -1/h
\end{bmatrix}
\end{equation*}
\begin{align*}
B &=
\begin{bmatrix} 0 & 0 & 1/h & 0 & 0 & 0 & 0 \end{bmatrix}^\top, \
E =
\begin{bmatrix} 0 & 0 & 0 & 0 & 0 & 0 & 1/h \end{bmatrix}^\top, \\
C &= \begin{bmatrix} 0 & 0 & 1 & 0 & 0 & 0 & 0 \end{bmatrix}.
\end{align*}
String stability is measured by the controller forcing signal $\chi_i$ 
which, after substituting~\eqref{eq:edot} and $\hat{u}_{i-1} = 
u_{i-1} + e_{u_{i-1}}$, can be written as
\begin{equation*}
    \chi_i = C_z\tilde{x}_i + D_z e_{u_{i-1}},
    \label{eq:chi_state}
\end{equation*}
with $C_z = \begin{bmatrix}      
    k_d  & \ 0 & \ 1 &\ k_p & \ -k_d(1-h\alpha) &\ -k_d h &\ 0 
    \end{bmatrix}$ and $D_z = 1$.

\subsection{Network layer}
To model the wireless V2V link, we consider event-triggered transmissions subject to stochastic communication
delays. In particular, transmissions are generated by a static event-triggering rule evaluated
at the sender, vehicle~$i-1$. An event is triggered when
\begin{equation*}
    \|e_{u_{i-1}}\|^2 \ge
    \sigma\Bigl[\rho\|u_{i-1}\|^2+
    \tfrac{1}{h^2}\|u_{i-1}-\chi_{i-1}\|^2\Bigr],
    \ \sigma,\rho>0,
\end{equation*}
where all quantities are locally available at vehicle~$i-1$
\citep[Remark~6]{Dolk2017}. A \emph{minimum inter-event time} (MIET) $\bar\tau_m>0$ enforces a minimum
time between transmissions, while a \emph{maximum allowable transmission interval} (MATI) $\bar\tau_s>\bar\tau_m$
guarantees that a transmission occurs at least every $\bar\tau_s$
seconds. Hence, the transmission instants satisfy
\begin{equation}\label{eq:tk}
    t_{k+1}^{i-1}:=
    \inf\Bigl\{t\ge t_k^{i-1}+\bar\tau_m \mid
    \Gamma(t)\le 0 \ \text{or}\ t-t_k^{i-1}=\bar\tau_s
    \Bigr\},
\end{equation}
where
\begin{equation}\label{eq:Gamma}
    \Gamma(t):=
    \sigma\Bigl[\rho\|u_{i-1}\|^2+
    h^{-2}\|u_{i-1}-\chi_{i-1}\|^2\Bigr]
    -\|e_{u_{i-1}}\|^2 .
\end{equation}
Unlike deterministic MAD-based approaches~\citep{Dolk2017,Heemels2010}, 
delays are drawn i.i.d.\ from a probability measure $\mu$ on $\R_{\ge 0}$,
which may assign positive probability to values exceeding the deterministic 
MAD. Specifically, at each transmission time $t_k^{i-1}$, a delay
$v_k^{i-1}\sim\mu$ is drawn and the received signal is updated as
\begin{equation}\label{eq:uhat_update}
    \hat u_{i-1}\bigl((t_k^{i-1}+v_k^{i-1})^+\bigr)
    =u_{i-1}(t_k^{i-1}),
\end{equation}
that is, vehicle~$i$ receives the value of $u_{i-1}$ sampled at 
the transmission instant $t_k^{i-1}$, delayed by $v_k^{i-1}$ seconds. 
Between updates, $\hat u_{i-1}$ is held constant.
To avoid packet reordering, we require the delay support to be bounded 
and smaller than the MATI, formalized as follows.
\begin{assumption}\label{ass:delay_dist}
    The delay distribution~$\mu$ has bounded support:
    $\mu([0,\bar v])=1$ for some $\bar v<\infty$, with $\bar v<\bar\tau_s$.\qed 
\end{assumption}
Finally, delays are assumed independent across links and independent of
the transmitted control signals.
\begin{assumption}\label{ass:independence}
    The delay sequences $\{v_k^{i-1}\}_{k\ge1}$ are independent across
    vehicle pairs $i=2,\ldots,N$, and independent of the predecessor's
    trajectory $\chi_{i-1}$.\qed 
\end{assumption}
\section{Hybrid systems modelling}
\label{sec:hybrid_model}
\begin{figure}[!t]
\centering
\begin{tikzpicture}[
  >=Stealth, thick,
  veh/.style={draw, fill=black!5, rectangle,
              minimum height=2.6em, minimum width=5em,
              font=\footnotesize, align=center, inner sep=3pt},
  comm/.style={draw, rectangle, minimum height=2em, minimum width=4em,
               font=\scriptsize, align=center, inner sep=2pt},
  dcomm/.style={comm, dashed},
  mode/.style={draw, circle, thick, minimum size=2.2em,
               font=\scriptsize, inner sep=0pt},
  lbl/.style={font=\scriptsize},
  sub/.style={font=\tiny, text=black!55},
]

\node[veh] (pred) at (0, 0) {Vehicle $i{-}1$};
\node[veh] (ego)  at (5.6, 0) {Vehicle $i$\\[-1pt]{\scriptsize (CACC)}};

\draw[->] (pred) -- node[above, lbl] {$\chi_{i-1}$} (ego);
\draw[->] (ego.east) -- ++(0.6, 0) node[right, lbl] {$\chi_i$};

\node[comm]  (et)  at (0, -1.8)   {\scriptsize Event\\[-2pt]\scriptsize trigger};
\node[dcomm] (ch)  at (2.8, -1.8) {\scriptsize Channel\\[0pt]\scriptsize $v_k^{i-1} \sim \mu$};
\node[comm]  (zoh) at (5.6, -1.8) {\scriptsize ZOH};

\draw[->] (pred.south) -- node[left, lbl] {$u_{i-1}$} (et.north);
\draw[->] (zoh.north) -- node[right, lbl] {$\hat{u}_{i-1}$} (ego.south);

\draw[->] (et) -- (ch);
\draw[->] (ch) -- (zoh);

\node[sub] at (3.9, -0.9) {$e_{u,i} = \hat{u}_{i-1} - u_{i-1}$};




\end{tikzpicture}
\caption{Per-vehicle-pair architecture. 
}
\label{fig:shs}
\end{figure}
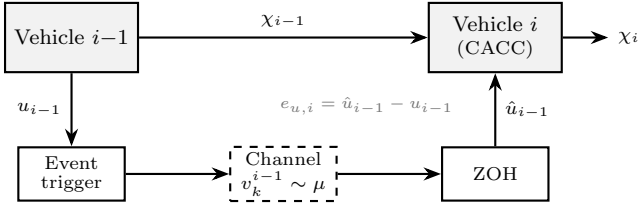

We adopt the SHS framework 
of~\cite{Teel2013,TeelHespanha2015,Schlotterbeck2024}. A stochastic 
hybrid system on~$\R^n$ is
\begin{equation}\label{eq:shs}
\mathcal{H}\colon \begin{cases}
  \dot{\xi} = F(\xi), & \xi \in \mathcal{C}, \\
  \xi^+ = G(\xi, v), & \xi \in \mathcal{D}, \; v \sim \mu,
\end{cases}
\end{equation}
with flow set~$\mathcal{C}$, jump set~$\mathcal{D}$, flow map~$F$, 
jump map~$G$, and a probability measure~$\mu$ on~$\R^m$ governing 
the random input~$v$ at each jump.

We model the platoon NCS for vehicle pair $(i{-}1, i)$ as a 
stochastic hybrid system with two operating modes, as illustrated 
in Fig.~\ref{fig:shs}. The augmented state is
\begin{align*}\label{eq:xi}
    \xi_i &:= \bigl(\xtilde_i,\, e_{u_{i-1}},\, \tau_{i-1},\, l_{i-1},\, \bar{d}_{i-1}\bigr) \in \mathbb{X}_i, \\
    \mathbb{X}_i &:= \mathbb{R}^7 \times \mathbb{R} \times \mathbb{R}_{\ge 0} \times \{0,1\} \times \mathbb{R}_{\ge 0}
\end{align*}
where $\tau_{i-1} \ge 0$ is a timer that counts time since the 
last transmission event, $l_{i-1} \in \{0,1\}$ is the operating 
mode, and $\bar{d}_{i-1} \ge 0$ is the delay value sampled 
from~$\mu$ at the most recent transmission jump 
($l_{i-1}{=}0 \to 1$) and stored in the state 
so that the SHS timer can track when the 
update~\eqref{eq:uhat_update} is due. The two modes are:
\begin{itemize}
  \item \textbf{Mode $l_{i-1}{=}0$ (inter-transmission):} the 
  system flows until a new transmission 
  is triggered according to ~\eqref{eq:tk}.
  \item \textbf{Mode $l_{i-1}{=}1$ (delay pending):} the system 
  flows until $\tau_{i-1}$ reaches~$\bar{d}_{i-1}$, at which point 
  the update~\eqref{eq:uhat_update} is applied.
\end{itemize}

The per-vehicle-pair SHS $\mathcal{H}_i$, driven by the predecessor 
output $\chi_{i-1}$ as an external input, takes the form
\begin{equation}\label{eq:shs_i}
\mathcal{H}_i\colon \begin{cases}
  \dot{\xi}_i = F_i(\xi_i, \chi_{i-1}), & \xi_i \in \mathcal{C}_i, \\
  \xi_i^+ = G_i(\xi_i, v^{i-1}), & \xi_i \in \mathcal{D}_i, \; 
  v^{i-1} \sim \mu,
\end{cases}
\end{equation}
where $v^{i-1}$ denotes the transmission delay on the 
link $(i{-}1)\to i$. 
The flow map $F_i$, jump map $G_i$, and sets $\mathcal{C}_i$, 
$\mathcal{D}_i$ are defined below.

\emph{Flow dynamics ($\xi_i \in \mathcal{C}_i$):} The flow map is
\begin{equation*}\label{eq:flowmap}
F_i(\xi_i, \chi_{i-1}) := \begin{pmatrix} A_{11}\xtilde_i + 
A_{12}e_{u_{i-1}} + A_{13}\chi_{i-1} \\ 
-\dot{u}_{i-1} \\ 1 \\ 0 \\ 0 \end{pmatrix}.
\end{equation*}
The flow set is
\begin{multline*}\label{eq:flowset}
\mathcal{C}_i := \bigl\{\xi_i \in \mathbb{X}_i \mid l_{i-1}{=}0 
\wedge \tau_{i-1} \le \bar{\tau}_s \wedge 
\bigl(\tau_{i-1} < \bar{\tau}_m \vee \\ \Gamma(\xi_i) \ge 0\bigr) 
\bigr\}
\cup\; \bigl\{\xi_i \in \mathbb{X}_i \mid l_{i-1}{=}1 \wedge 
\tau_{i-1} \le \bar{d}_{i-1}\bigr\}.
\end{multline*}
where $\Gamma(\xi_i)$ is the event-triggering function~\eqref{eq:Gamma}. Here we slightly abuse notation by expressing its dependence directly on the augmented state $\xi_i$ rather than on time $t$.

\emph{Jump dynamics ($\xi_i \in \mathcal{D}_i$):} The jump set is
\begin{multline*}
\mathcal{D}_i := \bigl\{\xi_i \in \mathbb{X}_i \mid l_{i-1}{=}0 
\wedge \tau_{i-1} \ge \bar{\tau}_m \wedge 
\bigl(\tau_{i-1} = \bar{\tau}_s \vee \\ \Gamma(\xi_i) \le 0\bigr)
\bigr\} 
\cup\; \bigl\{\xi_i \in \mathbb{X}_i \mid l_{i-1}{=}1 \wedge 
\tau_{i-1} = \bar{d}_{i-1}\bigr\},
\end{multline*}
with jump map
\begin{equation*}
G_i(\xi_i, v^{i-1}) := \begin{cases}
\bigl(\xtilde_i,\; e_{u_{i-1}},\; 0,\; 1,\; v^{i-1}\bigr)^\top 
& \text{if } l_{i-1} = 0, \\
\bigl(\xtilde_i,\; 0,\; \tau_{i-1},\; 0,\; 0\bigr)^\top 
& \text{if } l_{i-1} = 1,
\end{cases}
\end{equation*}
where $v^{i-1} \sim \mu$ is drawn at each transmission $l_{i-1}{=}0 \to 1$.

\begin{remark}
Zeno behavior is precluded: each update jump 
($l_{i-1}{=}1 \to 0$) resets 
$e_{u_{i-1}}^+ = 0$, $\tau_{i-1}^+ = 0$, $l_{i-1}^+ {=} 0$, 
and the MIET guard $\tau_{i-1} \ge \bar{\tau}_m$ prevents any 
transmission jump until $\tau_{i-1} = \bar{\tau}_m > 0$. Thus 
consecutive transmissions are separated by at least $\bar{\tau}_m$ 
seconds, regardless of the delay.
\end{remark}

\section{Main Results}
\label{sec:main}

\subsection{String stability notion}

Having modeled each vehicle pair as the SHS $\mathcal H_i$
in~\eqref{eq:shs_i}, we now introduce the stochastic counterpart of the
pairwise $\Ltwo$-gain condition used for deterministic string stability.
For the pair $(i{-}1,i)$, the predecessor signal $\chi_{i-1}$ is viewed
as the input and $\chi_i$ as the output, while the randomness enters
through the delay process on the communication link $(i{-}1)\to i$.

\begin{definition}[Expected $\Ltwo$ string stability]\label{def:ss_prob}
The hybrid system $\mathcal H_i$ in~\eqref{eq:shs_i} is said to be
expected $\Ltwo$-stable from input $\chi_{i-1}$ to output $\chi_i$
with gain $\theta$ if, for zero initial condition $\xi_i(0)=0$,
\begin{equation}\label{eq:ss_prob}
\E\!\left[\|\chi_i\|_{\Ltwo,[0,T]}^2 \mid \chi_{i-1}\right]
\le \theta^2 \|\chi_{i-1}\|_{\Ltwo,[0,T]}^2
\end{equation}
for all $T>0$ and every realization of $\chi_{i-1}$ with finite
$\Ltwo$ energy. The platoon is expected $\Ltwo$ string stable if this
property holds with $\theta \le 1$ for every vehicle pair $(i{-}1,i)$.
\end{definition}
The conditional expectation in~\eqref{eq:ss_prob} is taken with respect
to the local delay process only. Thus, for a fixed realization of the
predecessor signal $\chi_{i-1}$, the definition measures the expected
energy amplification induced by the stochastic communication channel.


\subsection{Sufficient conditions and main theorem}
We now provide sufficient conditions for the expected $\Ltwo$ string
stability property in Definition~\ref{def:ss_prob}. The certificate has
two components. The first is a quadratic dissipation inequality for the
continuous-time dynamics of each vehicle pair. The second is a
distributional condition ensuring that the Lyapunov function does not
increase in expectation across stochastic communication jumps.
\begin{condition}[String stability LMI]\label{cond:lmi}
There exist $P=P^T\succeq 0$, scalars $\nu>0$, $\gamma_l>0$, and
$\epsilon\ge 0$ such that
\begin{equation}\label{eq:lmi}
M :=
\begin{bmatrix}
M_{11} & M_{12} & M_{13} \\
\star & M_{22} & 0 \\
\star & \star & M_{33}
\end{bmatrix}
\preceq 0,
\end{equation}
where $M_{11} = \He(PA_{11}) + \nu C_z^T C_z
+ \Bigl(\rho+\tfrac{1}{h^2}\Bigr) C^T C$, $M_{12} = PA_{12} + \nu C_z^T D_z$, $M_{13} = PA_{13} - \tfrac{1}{h^2} C^T$, $M_{22} = \nu D_z^T D_z - \gamma_l^2$, $M_{33} = \tfrac{1}{h^2} - (1+\epsilon)\nu$. \qed
\end{condition}
Condition~\ref{cond:lmi} yields a dissipativity estimate with local
gain $\theta=\sqrt{1+\epsilon}$. Hence, to certify string stability in
the sense of Definition~\ref{def:ss_prob}, we will require the LMI to
hold with $\epsilon=0$, which gives $\theta=1$. The scalar $\gamma_l$
quantifies the effect of the network-induced error and determines the
timer-dependent comparison functions used to handle delays.

\begin{assumption}\label{ass:phi_functions}
Given the scalar $\gamma_l>0$ from Condition~\ref{cond:lmi}, there exist
design parameters $\gamma_0,\gamma_1>0$ and comparison functions
$\phi_0:[0,\bar{\tau}_s]\to(0,\infty)$ and
$\psi_1:[0,\bar v]\times[0,\bar v]\to(0,\infty)$ satisfying
\begin{align}
\dot{\phi}_0(\tau)
&= -\gamma_0\phi_0^2(\tau)-\gamma_l^2/\gamma_0,
\quad \phi_0(0)>0, \label{eq:phi_riccati}\\
\partial_\tau \psi_1(\tau,\bar d)
&= -\gamma_1\psi_1^2(\tau,\bar d)-\gamma_l^2/\gamma_1,
\quad \psi_1(0,\bar d)>0, \label{eq:psi1_riccati}
\end{align}
with $\phi_0(\tau)>0$ for all $\tau\in[0,\bar{\tau}_s]$ and
$\psi_1(\tau,\bar d)>0$ for all $\tau\in[0,\bar d]$.\qed 
\end{assumption}
The function $\phi_0$ weights the network-induced error during the
inter-transmission mode $l_{i-1}=0$, whereas $\psi_1$ weights the same
error during the delay-pending mode $l_{i-1}=1$. These Riccati
comparison functions are inspired by emulation-based networked control
analysis and are used here to compensate for the growth of the
network-induced error during flows, see e.g., \citep{Heemels2010,Schlotterbeck2024}. The solution of~\eqref{eq:phi_riccati} is
\begin{equation}\label{eq:riccati_solution}
\phi_0(\tau) = \tfrac{\gamma_l}{\gamma_0}\tan\!\bigl(\arctan(\phi_0(0)\gamma_0/\gamma_l) 
- \gamma_l\tau\bigr).
\end{equation}
Thus, $\phi_0$ decreases from $\phi_0(0)$ and remains positive on
$[0,\bar{\tau}_s]$ only if $\bar{\tau}_s$ is smaller than its first zero.
In the limiting case $\phi_0(0)\to\infty$, this gives the hard timer
bound $\bar{\tau}_s<\pi/(2\gamma_l)$. The same Riccati structure applies
to $\psi_1$ in the delay-pending mode, which yields the corresponding
support restriction $\bar v<\pi/(2\gamma_l)$ for the delay distribution.

\begin{condition}[Distributional delay condition]\label{cond:dist}
The MATI $\bar{\tau}_s$ and delay support bound $\bar v$ in Assumption~\ref{ass:delay_dist}
satisfy $\bar{\tau}_s,\bar v<\pi/(2\gamma_l)$. Moreover, the initial weighting satisfies
\begin{equation}\label{eq:dist_integral}
\gamma_1 \int_0^{\bar v} \psi_1(0,v)\,\mu(dv)
\le
\gamma_0\phi_0(\bar{\tau}_s).
\end{equation}
\end{condition}
Condition~\ref{cond:dist} is where the delay distribution enters the
certificate. It ensures that the Lyapunov function does not increase in
expectation when a transmission jump occurs and a random delay is drawn. Condition~\ref{cond:dist} still depends on the design choice
$\psi_1(0,\bar d)$. To obtain a directly verifiable test, we choose the smallest initial value that keeps $\psi_1(\tau,\bar d)$ positive on $[0,\bar d]$, namely $\psi_1(0,\bar d) = (\gamma_l/\gamma_1)\tan(\gamma_l\bar d)$. This choice maximizes the admissible delay range for each $\bar d$, since any larger $\psi_1(0,\bar d)$ would shrink the interval on which $\psi_1$ remains positive. With this choice, \eqref{eq:dist_integral} becomes
\begin{align}\label{eq:dist_tan}
\gamma_l\,\E\!\left[\tan(\gamma_l v)\right]
\le
\gamma_0\phi_0(\bar{\tau}_s).
\end{align}
Taking $\phi_0(0)\to\infty$ in~\eqref{eq:riccati_solution} gives $\gamma_0\phi_0(\bar{\tau}_s)
=
\frac{\gamma_l}{\tan(\gamma_l\bar{\tau}_s)}$, and therefore the directly verifiable sufficient condition
\begin{equation}\label{eq:cond_2_easy}
\E\!\left[\tan(\gamma_l v)\right]
\le
\frac{1}{\tan(\gamma_l\bar{\tau}_s)}.
\end{equation}
Unlike deterministic MAD conditions, this test depends on the full delay
distribution $\mu$, not only on its support. Distributions concentrated
near zero have small $\E[\tan(\gamma_l v)]$ and may satisfy
\eqref{eq:dist_tan} even when their support extends beyond the
deterministic MAD; more spread-out distributions, or distributions that
place significant probability mass near $\pi/(2\gamma_l)$, are harder to
certify.

\begin{remark}\label{rem:dist_shape}
Condition~\ref{cond:dist} is distribution-dependent in the sense that two
delay laws with the same support may lead to different feasibility
outcomes. This contrasts with deterministic worst-case MAD conditions \citep{Heemels2010,Dolk2017},
which depend only on an upper bound for the delay. The hard support
restriction $\bar v<\pi/(2\gamma_l)$ is a consequence of the Riccati
timer limit, whereas the integral condition~\eqref{eq:dist_integral}
captures how probability mass is distributed within that support.
\end{remark}

\ifFullVersion
\fi

\begin{theorem}[Expected $\Ltwo$ string stability]\label{thm:main}
Consider the SHS $\mathcal H_i$ in~\eqref{eq:shs_i} under
Assumptions~\ref{ass:delay_dist}--\ref{ass:phi_functions}. Suppose that
Condition~\ref{cond:lmi} holds with $\epsilon=0$ and that
Condition~\ref{cond:dist} holds. Then $\mathcal H_i$ is expected
$\Ltwo$-stable from $\chi_{i-1}$ to $\chi_i$ with gain at most one.
Consequently, the platoon is expected $\Ltwo$ string stable in the sense
of Definition~\ref{def:ss_prob}.
\end{theorem}
\ifFullVersion
    \emph{Proof.} See Appendix \ref{app:proof}.\hfill$\blacksquare$
\fi


\section{Numerical Examples}
\label{sec:numerical}
Theorem~\ref{thm:main} provides a 
\emph{sufficient} certificate: if the LMI and distributional condition hold, expected $\Ltwo$ string stability is guaranteed. The simulations illustrate how the certificate depends on the delay distribution and vehicle parameters, and validate the bounds via Monte Carlo. The LMI is solved using Python~3.11 with NumPy 
and CVXPY (SCS , \texttt{tol=1e-4}). The hybrid system is simulated using the HyEQ 
Toolbox~\citep{Sanfelice2013HyEQ}. Seeds are fixed (\texttt{seed=42}). 
We use the CACC parameters from~\cite[Sec.~VII]{Dolk2017}: $h = 0.6$\,s, $\tau_d = 0.1$\,s, $k_p = 0.2$, $k_d = 0.7$, $\rho = 0.04$, $r = 2.5$\,m, $L_i = 4$\,m. We set the velocity damping $\alpha = 0.1$ and $\epsilon = 0$. 
\subsection{The uniform distribution.}
Solving Condition~\ref{cond:lmi} yields $\gamma_l \approx 6.58$ and 
$\nu \approx 5.61$ at $\epsilon = 0$. The hard delay limit is 
$\pi/(2\gamma_l) \approx 238$\,ms, so both $\bar{\tau}_s$ and $\bar{v}$ 
must satisfy $\bar{\tau}_s, \bar{v} < 238$\,ms. We choose 
$\bar{\tau}_s = 200$\,ms close to the hard bound and 
$\bar{\tau}_m = 10$\,ms to preclude Zeno behaviour. With these parameters, the distributional threshold (RHS of \eqref{eq:cond_2_easy}) evaluates to approximately $0.26$.

First, we simulate the system under three uniform delay distributions 
$\mathcal{U}[0, \bar{v}]$ with increasing support: a \emph{light}, 
\emph{moderate}, and \emph{heavy} delay case. Table~\ref{tab:results_unif} 
summarises the distributional condition~\eqref{eq:cond_2_easy} for each 
scenario.
\begin{enumerate}
  \item \textbf{Light:} $\mu = \mathcal{U}[0, 55\,\text{ms}]$
  \item \textbf{Moderate:} $\mu = \mathcal{U}[0, 180\,\text{ms}]$
  \item \textbf{Heavy:} $\mu = \mathcal{U}[0, 500\,\text{ms}]$
\end{enumerate}
\begin{table}[!t]
\centering
\caption{Distributional feasibility for uniform delay cases. 
Hard limit: $238$\,ms.}
\label{tab:results_unif}
\setlength{\tabcolsep}{3pt}
\footnotesize
\begin{tabular}{@{}lcccc@{}}
\toprule
Scenario & $\bar{v}$ (ms) & $\E[v]$ (ms) & $\E[\tan(\gamma_l v)]$ & Feas. \\
\midrule
1. Light     & $55$  & $27.5$ & $0.19$  & \checkmark \\
2. Moderate  & $180$ & $90$   & $0.82$  & \xmark \\
3. Heavy     & $500$ & $250$  & $3328$  & \xmark \\
\midrule
\multicolumn{5}{@{}l}{LMI: $\gamma_l \!\approx\! 6.58$, $\epsilon = 0$. Threshold $(\tan(\gamma_l\bar{\tau}_s))^{-1}$: $0.26$.} \\
\bottomrule
\end{tabular}
\end{table}
As shown in Table~\ref{tab:results_unif}, only the Light scenario satisfies Condition~\eqref{eq:cond_2_easy}. The Moderate case violates the distributional condition, while the Heavy case also exceeds the support bound $\bar{v}<238$\,ms. Since $\tan(\gamma_l v)$ diverges as $\bar{v}\to \pi/(2\gamma_l)$, the corresponding integral becomes unbounded, explaining the large value in the Heavy case. Notably, the Light case allows delays up to $55$\,ms, more than twice the Dolk MAD of $26$\,ms~\cite[Sec.~VII]{Dolk2017}, showing that our framework certifies string stability in cases where the deterministic approach of~\cite{Dolk2017} would fail.

We simulate a platoon of $15$ follower vehicles under $N = 100$ Monte Carlo runs. As per Theorem \ref{thm:main}, we evaluate the mean $\mathcal{L}_2$-gain $\mathbb{E}[\|\chi_i\|^2/\|\chi_{i-1}\|^2]$ at each vehicle pair. Figure~\ref{fig:gain_vs_delay} shows the results for each distribution scenario.
The Light scenario satisfies Condition~\eqref{eq:cond_2_easy} and is empirically string stable. The Moderate case violates the condition but remains stable, illustrating the conservatism of the Lyapunov-based certificate. In contrast, the Heavy case exhibits string instability, confirming that sufficiently large delay supports can destabilize the platoon.

\begin{figure}[!t]
\centering
\includegraphics[width=\columnwidth]{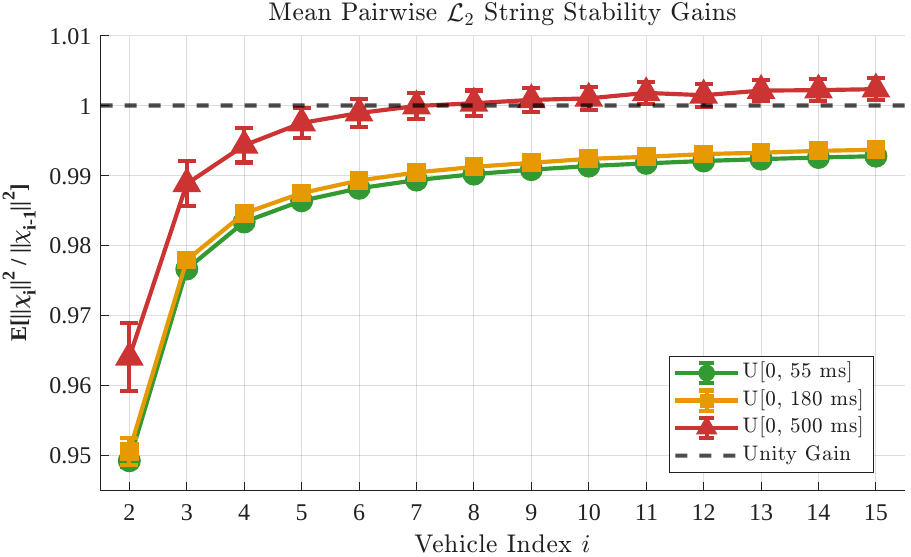}
\caption{Mean $\mathcal{L}_2$-gain $\mathbb{E}[\|\chi_i\|/\|\chi_{i-1}\|]$ 
per vehicle pair under $N = 100$ Monte Carlo runs for the scenarios of Table~\ref{tab:results_unif}.}
\label{fig:gain_vs_delay}
\end{figure}

String stability is further illustrated by the mean spacing error $e_i(t)$ over time, shown in Figure~\ref{fig:spacing_error_uniform}. In the certified Light scenario and the conservatively excluded Moderate scenario, the spacing error attenuates along the platoon, confirming string stability. In the Heavy scenario, the spacing error grows across vehicle pairs between $t = 25$ and $t = 40$ seconds, consistent with string instability.

\begin{figure}[!t]
\centering
\includegraphics[width=\columnwidth]{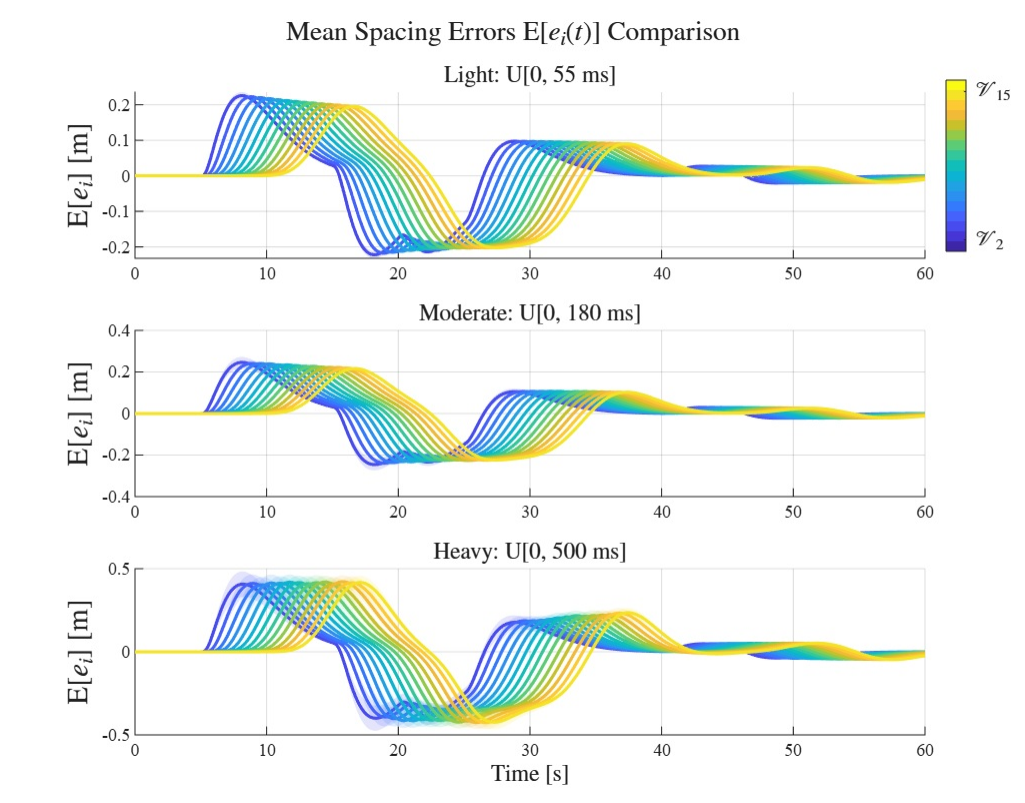}
\caption{Mean spacing error $e_i(t)$ for the scenarios of Table~\ref{tab:results_unif} ($N = 100$ Monte Carlo runs, $15$ follower vehicles).}
\label{fig:spacing_error_uniform}
\end{figure}

\subsection{Comparing different distributions.}
We now demonstrate that the certificate applies to any delay 
distribution satisfying Condition~\eqref{eq:cond_2_easy}, not only 
uniform ones. We evaluate two groups of distributions. The first group 
fixes $\bar{v} = 180$\,ms and considers three \emph{light} delay 
distributions: a truncated exponential $\mathrm{Exp}(\lambda = 28)$, 
a truncated gamma $\mathrm{Gamma}(k = 2,\, \theta = 0.018)$, and a 
Dirac point mass $\delta_{180}$ concentrating all delays at $180$\,ms. 
The second group fixes $\bar{v} = 500$\,ms and considers three 
\emph{heavy} delay distributions: a truncated exponential 
$\mathrm{Exp}(\lambda = 10)$, a truncated gamma 
$\mathrm{Gamma}(k = 2,\, \theta = 0.3)$, and a Dirac point mass 
$\delta_{500}$ concentrating all delays at $500$\,ms. 
Table~\ref{tab:results_ss_various_dist} reports 
Condition~\eqref{eq:cond_2_easy} for each case.

\begin{table}[!t]
\centering
\caption{Distributional feasibility for various delay distributions.
Hard limit: $238$\,ms.}
\label{tab:results_ss_various_dist}
\setlength{\tabcolsep}{3pt}
\footnotesize
\begin{tabular}{@{}lcccc@{}}
\toprule
Scenario & $\bar{v}$ (ms) & $\E[v]$ (ms) & $\E[\tan(\gamma_l v)]$ & Feas. \\
\midrule
1. L. $\mathrm{Exp}(\lambda = 28)$ & $180$ & $34$ & $0.25$ & \checkmark \\
2. L. $\mathrm{Gamma}(k = 2,\, \theta = 0.018)$ & $180$ & $36$ & $0.25$ & \checkmark \\
3. L. $\delta_{180}$ & $180$ & $180$ & $2.46$ & \xmark \\
4. H. $\mathrm{Exp}(\lambda = 10)$ & $500$ & $96$ & $536$ & \xmark \\
5. H. $\mathrm{Gamma}(k = 2,\, \theta = 0.3)$ & $500$ & $283$ & $3930$ & \xmark \\
6. H. $\delta_{500}$ & $500$ & $500$ & $6366$ & \xmark \\
\midrule
\multicolumn{5}{@{}l}{LMI: $\gamma_l \!\approx\! 6.58$, $\epsilon = 0$. Threshold $(\tan(\gamma_l\bar{\tau}_s))^{-1}$: $0.26$.} \\
\bottomrule
\end{tabular}
\end{table}

The light truncated exponential and gamma satisfy Condition~\eqref{eq:cond_2_easy} and are certified string stable, while the remaining four do not. This highlights a key feature of the certificate: two distributions with the same support can have very different feasibility depending on how mass concentrates near zero, allowing the framework to certify string stability even when delays exceed the deterministic MAD of~\cite{Dolk2017} with nonzero probability.

Figure~\ref{fig:gain_vs_delay_different-distribution-ss} reports the mean pairwise $\mathcal{L}_2$-gain over $N=100$ Monte Carlo runs. The light truncated exponential and gamma cases are empirically string stable, in agreement with the certificate. The light point mass $\delta_{180}$ and heavy exponential are also stable, revealing the conservatism of the sufficient condition. In contrast, the heavy gamma and heavy point mass $\delta_{500}$ become string unstable, with gains exceeding unity for some vehicle pairs. Since the heavy exponential remains stable despite sharing the same support $\bar{v}=500$\,ms as the unstable cases, the results show that string stability depends on the full delay distribution, not only on its support.

\begin{figure}[!t]
\centering
\includegraphics[width=\columnwidth]{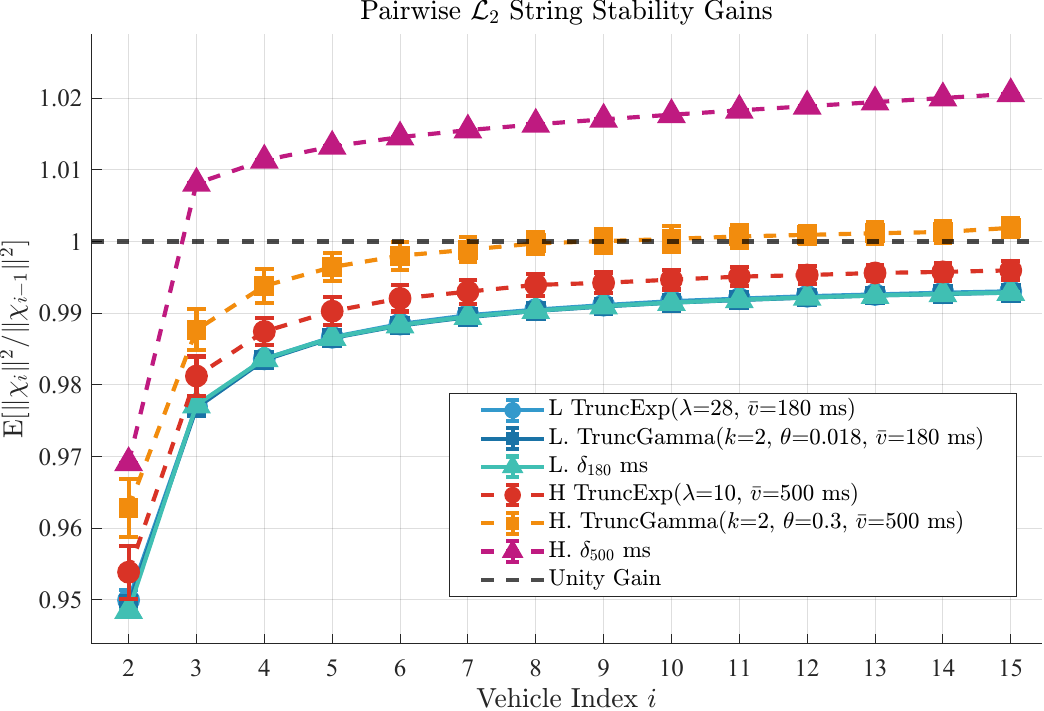}
\caption{Mean $\mathcal{L}_2$-gain per vehicle pair under $N = 100$ 
Monte Carlo runs for the delay distributions of 
Table~\ref{tab:results_ss_various_dist}.}
\label{fig:gain_vs_delay_different-distribution-ss}
\end{figure}

%
\section{Conclusion}
\label{sec:conclusion}
This paper studied expected $\mathcal{L}_2$ string stability of event-triggered CACC platoons with a human-driven leader under stochastic communication delays. We derived distribution-dependent sufficient conditions, expressed in terms of the full delay law, that certify string stability even when delays exceed deterministic admissible bounds with nonzero probability. HyEQ simulations validated the results and showed that feasibility depends on the shape of the delay distribution, not only on its support. Future work will focus on reducing the conservatism of the Lyapunov-based certificate, extending the framework to delay-adaptive event-triggering, and incorporating explicit models of human driver behavior at the leader to characterize the class of admissible driving maneuvers.

\ifFullVersion
  
\appendix

\section{Proof of Theorem~\ref{thm:main}}
\label{app:proof}
We first need the following preliminary result:
\begin{lemma}\label{lem:lyap_ugasp}
Consider an SHS~\eqref{eq:shs} with external input $w$ and output $z$,
so that $F=F(\xi,w)$ and $z=h(\xi,w)$. Let
$U:\R^n\to\R_{\ge 0}$ and define the supply rate
\[
s(w,z):=\nu\bigl[\theta^2\|w\|^2-\|z\|^2\bigr],
\]
where $\nu>0$ and $\theta\ge 0$. Suppose that, for every fixed input
$w\in\Ltwo[0,T]$,
\begin{align}
\langle \nabla U(\xi),F(\xi,w)\rangle
&\le s(w,z), \qquad \forall \xi\in\mathcal C, \label{eq:lyap_flow}\\
\E_v\!\left[U(G(\xi,v))\right]
&\le U(\xi), \qquad \forall \xi\in\mathcal D. \label{eq:lyap_jump}
\end{align}
Then, for every $T>0$,
\[
\E_v\!\left[\|z\|_{\Ltwo,[0,T]}^2\right]
\le
\theta^2\|w\|_{\Ltwo,[0,T]}^2
+
\frac{U(\xi(0))}{\nu},
\]
where the expectation is taken only with respect to the random jumps of
the SHS.
\end{lemma}
%
%
\begin{pf}
Let $(\mathcal{F}_t)_{t\ge 0}$ be the natural filtration generated by the
i.i.d.\ delay sequence $\{v_k\}$ (Assumption~\ref{ass:independence}). Each
jump time $t_k$ is a stopping time, the pre-jump state $\xi(t_k^-)$ is
$\mathcal{F}_{t_k^-}$-measurable, and the draw $v_k$ is independent of
$\mathcal{F}_{t_k^-}$. Hence the jump bound~\eqref{eq:lyap_jump} holds in
the adapted sense,
$\E\!\left[U(\xi(t_k^+)) \mid \mathcal{F}_{t_k^-}\right] \le U(\xi(t_k^-))$.
Since Zeno behavior
is precluded, the number of jumps $K$ on $[0,T]$, $0 < t_1 < \cdots < t_K < T$, is finite almost surely.

 On each flow interval $[t_j, t_{j+1})$ (with $t_0 := 0$, $t_{K+1} := T$), integrating~\eqref{eq:lyap_flow} gives
$U(\xi(t_{j+1}^-)) - U(\xi(t_j^+)) \le \nu \int_{t_j}^{t_{j+1}} \bigl[\theta^2\|w\|^2 - \|z\|^2\bigr] dt$.
Summing over all $K{+}1$ flow intervals yields a telescoping sum:
$U(\xi(T)) - U(\xi(0)) + \sum_{j=1}^{K} \bigl[U(\xi(t_j^-)) - U(\xi(t_j^+))\bigr] \le \nu \int_0^T \bigl[\theta^2\|w\|^2 - \|z\|^2\bigr] dt$.
Taking expectations and using~\eqref{eq:lyap_jump} at each jump, the tower property gives the unconditional inequality
$\E[U(\xi(t_j^+))] \le \E[U(\xi(t_j^-))]$, so each term $\E[U(\xi(t_j^-)) - U(\xi(t_j^+))] \ge 0$. Dropping these non-negative terms and using $\E[U(\xi(T))] \ge 0$ gives $-U(\xi(0)) \le \nu\theta^2\,\|w\|_{\Ltwo,[0,T]}^2 - \nu\,\E[\|z\|_{\Ltwo,[0,T]}^2]$, which rearranges to the stated bound.\qed 
\end{pf}

The proof employs the following composite Lyapunov function
\begin{multline}\label{eq:lyap_composite}
U(\xi_i) = \xtilde_i^T P \xtilde_i + 
(1{-}l_{i-1})\,\gamma_0\,\phi_0(\tau_{i-1})\,\|e_{u_{i-1}}\|^2 \\
+ l_{i-1}\,\gamma_1\,\psi_1(\tau_{i-1}, \bar{d}_{i-1})\,
\|e_{u_{i-1}}\|^2,
\end{multline}
which merges the quadratic $\Ltwo$-gain certificate 
$V(\xtilde_i) = \xtilde_i^T P \xtilde_i$ from~\cite{Dolk2017} with 
the comparison-function-weighted error terms 
of~\cite{Schlotterbeck2024}. The mode~$l_{i-1}{=}1$ 
weighting~$\psi_1(\tau_{i-1}, \bar{d}_{i-1})$ depends on the drawn 
delay~$\bar{d}_{i-1}$, enabling the distributional 
condition~\eqref{eq:dist_integral}. 

We now verify the two conditions of Lemma~\ref{lem:lyap_ugasp}: flow dissipation~\eqref{eq:lyap_flow} and expected jump non-increase~\eqref{eq:lyap_jump}.

\emph{Step~1 (Flow analysis).}
We treat both modes uniformly. For mode $l_{i-1}{=}k$ ($k \in \{0,1\}$), write $\varphi_k$ for the active comparison function ($\phi_0$ in mode~$0$, $\psi_1(\cdot, \bar{d})$ in mode~$1$). The active Lyapunov terms are $V = \xtilde_i^T P \xtilde_i$ and $\gamma_k\varphi_k(\tau)\|e_{u_{i-1}}\|^2$. Computing $\dot{U}$:
\begin{equation}\label{eq:Udot_general}
\dot{U} = 2\xtilde_i^T P\dot{\xtilde}_i + \gamma_k\bigl[\dot{\varphi}_k(\tau)\|e_{u_{i-1}}\|^2 + 2\varphi_k(\tau)\,e_{u_{i-1}}^T \dot{e}_{u_{i-1}}\bigr].
\end{equation}
The LMI~\eqref{eq:lmi} certifies the following dissipation inequality
for $V(\xtilde_i)=\xtilde_i^T P\xtilde_i$. Defining
$\theta^2:=1+\epsilon$, we have
\begin{multline}\label{eq:Vdot_bound}
\dot V =
2\xtilde_i^T P\dot{\xtilde}_i
\le
-\rho\|u_{i-1}\|^2
-\tfrac{1}{h^2}\|u_{i-1}-\chi_{i-1}\|^2 \\
+ \nu\bigl[\theta^2\|\chi_{i-1}\|^2-\|\chi_i\|^2\bigr]
+\gamma_l^2\|e_{u_{i-1}}\|^2,
\end{multline}
where $u_{i-1}=C\xtilde_i$ and
$\chi_i=C_z\xtilde_i+D_z e_{u_{i-1}}$.
Since $e_{u_{i-1}} \in \R$ is scalar and $\dot{e}_{u_{i-1}} = (1/h)(u_{i-1} {-} \chi_{i-1})$, the cross-term satisfies $|e_{u_{i-1}}\dot{e}_{u_{i-1}}| = |e_{u_{i-1}}| \cdot |\dot{e}_{u_{i-1}}|$ exactly. Defining $H := h^{-1}|u_{i-1} {-} \chi_{i-1}|$ and $W := |e_{u_{i-1}}|$, we complete the square:
\begin{equation*}
-H^2 + 2\gamma_k\varphi_k W H = -\bigl(H - \gamma_k\varphi_k W\bigr)^2 + \gamma_k^2\varphi_k^2 W^2.
\end{equation*}
Substituting into~\eqref{eq:Udot_general} and using~\eqref{eq:Vdot_bound}:
\begin{multline*}
\dot{U} \le -\bigl(H {-} \gamma_k\varphi_k W\bigr)^2 \\
 + \bigl(\gamma_k^2\varphi_k^2 + \gamma_l^2 + \gamma_k\dot{\varphi}_k\bigr)\|e_{u_{i-1}}\|^2 + s(w,z) - \rho\|u_{i-1}\|^2.
\end{multline*}
The Riccati ODE~\eqref{eq:phi_riccati}--\eqref{eq:psi1_riccati} is designed so that the $\|e_{u_{i-1}}\|^2$ coefficient vanishes identically:
\begin{equation*}
\gamma_k^2\varphi_k^2 + \gamma_l^2 + \gamma_k\dot{\varphi}_k = \gamma_k^2\varphi_k^2 + \gamma_l^2 + \gamma_k(-\gamma_k\varphi_k^2 - \gamma_l^2/\gamma_k) = 0.
\end{equation*}
The combined flow bound in both modes is therefore
\begin{multline*}
\dot{U} \le
-\bigl(H {-} \gamma_k\varphi_k W\bigr)^2
+ \nu\bigl[\theta^2\|\chi_{i-1}\|^2 - \|\chi_i\|^2\bigr] \\
-\rho\|u_{i-1}\|^2
\le s(w,z),
\end{multline*}
verifying condition~\eqref{eq:lyap_flow} of Lemma~\ref{lem:lyap_ugasp}.

\emph{Step~2 (Jump analysis).}

\emph{Type~1 jump ($l_{i-1}{=}0 \to l_{i-1}{=}1$):}
The delay $v \sim \mu$ is drawn, $\bar{d}^+ = v$, $\tau^+ = 0$, and
$l_{i-1}^+ = 1$. Now $U$ after the jump depends on~$v$ through
$\psi_1(0, v)$:
\begin{align*}
\E_v\!\bigl[U(\xi_i^+)\bigr] 
&= V(\xtilde_i) 
+ \gamma_1\,\E_v\!\bigl[\psi_1(0, v)\bigr]\,
\|e_{u_{i-1}}\|^2,\\
U(\xi_i^-) 
&= V(\xtilde_i) 
+ \gamma_0\,\phi_0(\tau)\,\|e_{u_{i-1}}\|^2. \nonumber
\end{align*}
At a Type~1 jump, the jump guard implies $\tau\in[0,\bar{\tau}_s]$.
Since $\phi_0$ is decreasing on this interval, we have
$\phi_0(\bar{\tau}_s)\le \phi_0(\tau)$. Hence, by~\eqref{eq:dist_integral}, $\gamma_1\E_v[\psi_1(0,v)]
\le \gamma_0\phi_0(\bar{\tau}_s)
\le \gamma_0\phi_0(\tau)$, and therefore $\E_v[U(\xi_i^+)]\le U(\xi_i^-)$. This is the step where
the delay distribution enters the analysis.

\emph{Type~2 jump ($l_{i-1}{=}1 \to l_{i-1}{=}0$, at $\tau = \bar{d}$):}
This jump is deterministic (no random draw). Since $e_{u_{i-1}}^+ = 0$ and $l_{i-1}^+ = 0$:
\begin{equation*}
U(\xi_i^+) = V(\xtilde_i) \le V(\xtilde_i) + \gamma_1\psi_1(\bar{d},\bar{d})\|e_{u_{i-1}}(\bar{d})\|^2 = U(\xi_i^-).
\end{equation*}
Both jump types satisfy condition~\eqref{eq:lyap_jump}.

\emph{Step~3 (Conclusion).}
Fix an arbitrary realization of the predecessor signal
$\chi_{i-1}$ such that $\chi_{i-1}\in\Ltwo[0,T]$. By
Assumption~\ref{ass:independence}, the delay process on the link
$(i{-}1)\to i$ is independent of this realized trajectory. Hence, for
this fixed input, the expectation in Lemma~\ref{lem:lyap_ugasp} is taken
only with respect to the local delay process.

Steps~1 and~2 verify the hypotheses of Lemma~\ref{lem:lyap_ugasp}, with
$w=\chi_{i-1}$ and $z=\chi_i$. Therefore,
\[
\E\!\left[\|\chi_i\|_{\Ltwo,[0,T]}^2 \mid \chi_{i-1}\right]
\le
\theta^2\|\chi_{i-1}\|_{\Ltwo,[0,T]}^2
+
\frac{U(\xi_i(0))}{\nu},
\]
where $\theta^2=1+\epsilon$. If Condition~\ref{cond:lmi} holds with
$\epsilon=0$, then $\theta=1$. For zero initial conditions,
$U(\xi_i(0))=0$, and hence
\[
\E\!\left[\|\chi_i\|_{\Ltwo,[0,T]}^2 \mid \chi_{i-1}\right]
\le
\|\chi_{i-1}\|_{\Ltwo,[0,T]}^2.
\]
Thus, $\mathcal H_i$ is expected $\Ltwo$-stable from $\chi_{i-1}$ to
$\chi_i$ with gain at most one. Since the argument holds for every
vehicle pair $(i{-}1,i)$, the platoon is expected $\Ltwo$ string stable.
\hfill$\blacksquare$

\fi

\bibliography{Francisco/references}

\end{document}